%% file: paper.tex
\documentclass{article}
\usepackage[paper=letterpaper,margin=1in]{geometry}

\usepackage{graphicx,yhmath}
\usepackage{pstricks,pstricks-add}
\usepackage{pst-node}
\usepackage{pst-coil}
\usepackage{pst-eucl}
\usepackage{amsthm}
\usepackage{amsmath}
\usepackage{amssymb}
\usepackage{amsfonts}
\usepackage{epsfig}
\usepackage{verbatim}
\usepackage{enumitem}
\usepackage{arcs}
\usepackage{pbox}
\usepackage{multirow}

\newtheorem{fact}{Fact}

\newtheorem{theorem}{Theorem}
\newtheorem{proposition}{Proposition}
\newtheorem{lemma}{Lemma}

\newlength{\alginputwidth}
\newlength{\algboxwidth}

\newsavebox{\algbox}
\newsavebox{\captionbox}
    {
        \setlength{\algboxwidth}{\columnwidth}
        \addtolength{\algboxwidth}{-\columnsep}
        \addtolength{\algboxwidth}{-1mm}
        \setlength{\alginputwidth}{\algboxwidth}
        \addtolength{\alginputwidth}{-1.7cm}
        \begin{figure}[htbp]
            \vspace*{-1mm}
            \centering
            \begin{lrbox}{\captionbox}
                \begin{minipage}[b]{\algboxwidth}
                    \centering
                    \label{#2}
                \end{minipage}
            \end{lrbox}
            \begin{lrbox}{\algbox}
                \begin{minipage}[b]{\algboxwidth}
                    \footnotesize
                    \vspace*{2mm}
    } 
    {
                    \vspace*{0.2mm}
               \end{minipage}
            \end{lrbox}
            \fbox{\usebox{\algbox}\hspace*{1mm}}
            \usebox{\captionbox}
            \vspace*{-1mm}
        \end{figure}
    }
\newsavebox{\algcodebox}
    {
        \begin{enumerate}
            \setlength{\itemsep}{2pt}
            \setlength{\parsep}{0pt}
            \setlength{\topsep}{0pt}
            \setlength{\parskip}{0pt}
            \setlength{\partopsep}{0pt}
    } 
    {\end{enumerate}}

\begin{document}

\title{The Yao Graph $Y_5$ is a Spanner}

\author{
{Wah Loon Keng}\thanks{
Lafayette College,
Easton, PA 18042, USA.
kengw{\tt @}lafayette.edu.} \and
{Ge Xia}\thanks{
Department of Computer Science,
Lafayette College,
Easton, PA 18042, USA.
xiag{\tt @}lafayette.edu.} \\
}
\date{}
\maketitle

\begin{abstract}
In this paper we prove that $Y_5$, the Yao graph with five cones, is a spanner with stretch factor $\rho = 2+\sqrt{3} \approx 3.74$. Since $Y_5$ is the only Yao graph whose status of being a spanner or not was open, this completes the picture of the Yao graphs that are spanners: a Yao graph $Y_k$ is a spanner if and only if $k \geq 4$.

We complement the above result with a lower bound of 2.87 on the stretch factor of $Y_5$. We also show that $YY_5$, the Yao-Yao graph with five cones, is not a spanner.
\vspace*{5mm}
\end{abstract}




\section{Introduction} \label{intro}
Let $S$ be a set of points in the plane. Fix an ordering $\prec$ on all pairs of points $\{a,b\}$ in $S$ based on their Euclidean distance $||ab||$ where ties are broken arbitrarily, i.e. if $||ab|| < ||cd||$ then $\{a,b\} \prec \{c,d\}$. Given an integer parameter $k > 0$, the {\em directed Yao graph}~\cite{yao}
with parameter $k$, denoted $\overrightarrow{Y_k}$, is constructed as follows. For each point $p$ in $S$, partition the space into $k$ equal-measured cones of angle $2\pi/k$ each whose apex is $p$ (the orientation of the cones is fixed for all points). In each cone, $p$ chooses the closest point $q$ in $S$ (if any) according to the ordering $\prec$ and adds $(p,q)$ to $\overrightarrow{Y_k}$ as a directed edge outgoing from $p$.
The (undirected) Yao graph with parameter $k$, denoted $Y_k$, is the underlying undirected graph of $\overrightarrow{Y_k}$.

A geometric graph $G$ on the point set $S$ is called a  {\em $\rho$-spanner} if for every two points $a,b \in S$, the shortest path distance between $a$ and $b$ in $G$ is at most $\rho\cdot||ab||$. $G$ is called a {\em geometric spanner} or simply {\em spanner} if $\rho$ is a constant.

The Yao graphs have been extensively studied, and in particular many of their spanning properties have been discovered. It is known that $Y_2$ and $Y_3$ are not spanners~\cite{molla}, $Y_4$ is a spanner with stretch factor $8\sqrt{2}(29+23\sqrt{2})$~\cite{y4}, $Y_6$ is a spanner with stretch factor $17.7$~\cite{y6}, and that for $k \geq 7$, $Y_k$ is a spanner with stretch factor $\frac{1}{1-2\sin(\pi/k))}$~\cite{y7}. The question of whether or not $Y_5$ is a spanner was previously open.


In this paper we prove that $Y_5$ is a $\rho$-spanner, where $\rho = 2+\sqrt{3} \approx 3.74$. Combining this with the previous results, we now have a complete picture of the spanners that can be constructed with Yao graphs: any Yao graph $Y_k$ is a spanner if and only if $k \geq 4$. We also give a lower bound of 2.87 on the stretch factor of $Y_5$.

\vspace*{.4cm}\noindent
{\bf Recent Developments.}
An earlier version of this paper~\cite{early} proved a stretch factor of $\frac{1}{1-2\sin{(3\pi/20)}} \approx 10.87$ for $Y_5$ using a simple approach. In a recent manuscript, Barba et al.~\cite{personal} independently proved the same bound of 10.87 using the same approach and they also used that approach to improve the stretch factor of $Y_k$ for odd $k\geq 7$ to $\frac{1}{1-2\sin(3\pi/4k)}$. In addition, Barba et al.~\cite{personal} improved the stretch factor of $Y_6$ to 5.8.

In contrast to our main results, we show that $YY_5$, the Yao-Yao graph with five cones, is not a spanner. The {\em directed Yao-Yao graph} with parameter $k > 0$, denoted $\overrightarrow{YY_k}$, is constructed in two stages. The first stage proceeds as in the construction of $\overrightarrow{Y_k}$. In the second stage, for each point $p \in S$, and for each cone defined by $p$ in the first stage, point $p$ keeps {\em only} the shortest incoming edge (if any) according to the ordering $\prec$ in $\overrightarrow{Y_k}$ in the cone. The directed edges kept by the points in $S$ in the second stage constitute $\overrightarrow{YY_k}$. The (undirected) Yao-Yao graph $YY_k$ denotes the underlying undirected graph of $\overrightarrow{YY_k}$. Clearly, $\overrightarrow{YY_k}$ is a subgraph of $\overrightarrow{Y_k}$, and $YY_k$ is a subgraph of $Y_k$. The Yao-Yao graphs have an advantage over the Yao graphs in that their maximum degree is bounded:  Whereas $Y_k$ can have unbounded degree, the maximum degree of $YY_k$ is at most $2k$. It is known that $YY_4$ is not a spanner~\cite{eucg} and is not plane~\cite{joco} and that for any integer $k \geq 6$, $YY_{6k}$ is a spanner~\cite{soda13}. It is still open whether the Yao-Yao graph is a spanner for other values of the parameter $k$.

Table~\ref{tab:results} shows the stretch factors of Yao and Yao-Yao graphs for various values of the parameter $k$.

\begin{table}
\caption {Stretch factors of Yao and Yao-Yao graphs} \label{tab:results}
\begin{center}\renewcommand{\arraystretch}{1.6}
\begin{tabular}{|c|c|c|}
  \hline
  Parameter $k$ & Yao Graph $Y_k$ & Yao-Yao Graph $YY_k$ \\\hline\hline
  $k=2,3$ & not a spanner~\cite{molla} & not a spanner~\cite{molla} \\\hline
  $k=4$ & $8\sqrt{2}(29+23\sqrt{2})$~\cite{y4} & not a spanner~\cite{eucg} \\\hline
  $k=5$ & \pbox{7.5cm}{$1/(1-2\sin{(3\pi/20)})\approx 10.87$~\cite{personal, early} \\ {\color{blue} $2+\sqrt{3} \approx 3.74$ [this paper]}} & {\color{blue} not a spanner [this paper]} \\[0.6ex]\hline
  $k=6$ & 17.7~\cite{y6}, 5.8~\cite{personal} & not a spanner~\cite{molla} \\\hline
  $k \geq 7$ & \pbox{7.5cm}{$1/(1-2\sin(\pi/k))$ for even $k$~\cite{y7} \\  $1/(1-2\sin(3\pi/4k))$ for odd $k$~\cite{personal}} & \pbox{5.6cm}{11.67 for $k=6k'$, $k'\geq 6$~\cite{soda13}\\ open for other values of $k\geq 7$} \\[0.6ex]
  \hline
\end{tabular}
\end{center}
\end{table}

The paper is organized as follows. In Section~\ref{sec:prelim}, we introduce the notations and terminologies used throughout the paper. In Section~\ref{section:main}, we prove that $Y_5$ is a spanner. In Section~\ref{section:lower_bound}, we give a lower bound of 2.87 on the stretch factor of $Y_5$. We show in Section~\ref{section:yy} that $YY_5$ is not a spanner. We conclude the paper in Section~\ref{sec:conclusion}.

\section{Preliminaries}
\label{sec:prelim}
Given a set of points $S$ in the two-dimensional Euclidean
plane, the complete Euclidean graph ${\cal E}$ on $S$ is
defined to be the complete graph whose point-set is $S$. Each
edge $ab$ connecting points $a$ and $b$ is assumed to be embedded in
the plane as the straight line segment $ab$; we define its {\em
length} to be the Euclidean distance $||ab||$.

Let $G$ be a subgraph of ${\cal E}$. The length of a simple path $P
= m_0, m_1,\ldots, m_r = b$ between two points $a, b$ in $G$ is $|P| = \sum_{j=0}^{r-1} ||m_jm_{j+1}||.$
For two points $a$, $b$ in $G$, we denote by $d_G(a, b)$ (or simply $d(a, b)$ if $G$ is clear from the context) the length of a shortest path between $a$ and $b$ in $G$. $G$ is said to be a {\em spanner} (of ${\cal E}$) if there is a constant $\rho$ such that, for every
two points $a,b \in G$, $d(a, b) \leq \rho \cdot ||ab||$. The constant $\rho$ is called the {\em
stretch factor} or {\em
spanning ratio} of $G$ (with respect to ${\cal E}$).

For each point $p\in S$, label the five cones around it by $C_1^p, C_2^p, \ldots, C_5^p$ in the counterclockwise order. The two rays on the boundary of each cone are referred to as the {\em start-ray} and the {\em end-ray}, in the counterclockwise order. Fix an orientation of the cones such that the start-ray of $C_1^p$ for all $p$ is horizontal and points to the right. The {\em bisector} of a cone is a ray that separates the cone into two equal-sized subcones. See Figure~\ref{fig:cones} for an illustration. The following is a simple fact:
\begin{fact}\label{fact2}
Rotating around any point in the plane by $2\pi n/5$, where $n$ is an integer, does not change the orientation of the cones (up to a relabeling). Furthermore, mirror-flipping along the bisector of any cone does not change the orientation of the cones (up to a relabeling).
\end{fact}

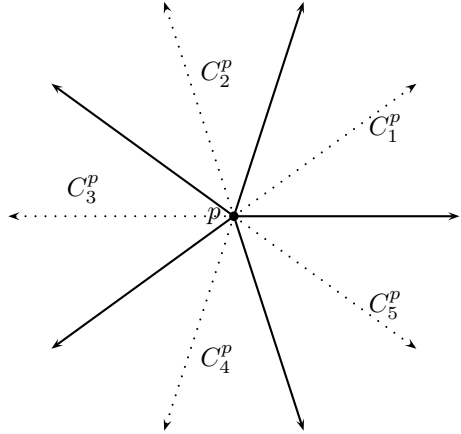
\begin{figure}[tbhp]
\begin{center}
\begin{pspicture}(6,6)
  \pnode(3,3){u}\psdot(u)\uput[180](u){$p$}

  \pnode(6,3){x1}

  \pnode([nodesep=3,angle=72]{x1}u){x2}
  \pnode([nodesep=3,angle=72]{x2}u){x3}
  \pnode([nodesep=3,angle=72]{x3}u){x4}
  \pnode([nodesep=3,angle=72]{x4}u){x5}

  \psline{->}(u)(x1)
  \psline{->}(u)(x2)
  \psline{->}(u)(x3)
  \psline{->}(u)(x4)
  \psline{->}(u)(x5)

  \pnode([nodesep=2,angle=36]{x1}u){y1}\uput[0](y1){$C_1^p$}
  \pnode([nodesep=3,angle=36]{x1}u){z1}\psline[linestyle=dotted]{->}(u)(z1)

  \pnode([nodesep=2,angle=36]{x2}u){y2}\uput[0](y2){$C_2^p$}
  \pnode([nodesep=3,angle=36]{x2}u){z2}\psline[linestyle=dotted]{->}(u)(z2)

  \pnode([nodesep=2,angle=36]{x3}u){y3}\uput[90](y3){$C_3^p$}
  \pnode([nodesep=3,angle=36]{x3}u){z3}\psline[linestyle=dotted]{->}(u)(z3)

  \pnode([nodesep=2,angle=36]{x4}u){y4}\uput[0](y4){$C_4^p$}
  \pnode([nodesep=3,angle=36]{x4}u){z4}\psline[linestyle=dotted]{->}(u)(z4)

  \pnode([nodesep=2,angle=36]{x5}u){y5}\uput[0](y5){$C_5^p$}
  \pnode([nodesep=3,angle=36]{x5}u){z5}\psline[linestyle=dotted]{->}(u)(z5)

\end{pspicture}
\caption{The cones and their bisectors.}\label{fig:cones}
\end{center}
\end{figure}

In this paper, all the angles labeled as $\angle xyz$ are measured from ray $\overrightarrow{yx}$ to ray $\overrightarrow{yz}$ in counterclockwise direction. $|\angle xyz|$ indicates the (unsigned) magnitude of  $\angle xyz$.

Next we give two lemmas that will be useful in our proof.

\begin{lemma}\label{lem:bdd}
Let $a$, $b$, and $c$ be three distinct points in the plane such that $||ac|| \leq ||ab||$ and $|\angle bac| \leq \theta$, where $\theta \in (0, \pi/3)$ is a constant. Then
$$||ac||+\lambda||bc|| \leq \lambda||ab||,$$ where $\lambda =\frac{1}{1-2\sin (\theta/2)}$.
\end{lemma}
\begin{proof}
By Lemma 10 of~\cite{y7}, $
    ||bc|| \leq ||ab||-||ac||/t,$
    where $t=\frac{1+\sqrt{2-2\cos \theta}}{2\cos \theta -1}.$ By trigonometric identities, $
    t=\frac{1+\sqrt{2-2\cos \theta}}{2\cos \theta -1}=\frac{1}{1-\sqrt{2-2\cos \theta}}=\frac{1}{1-\sqrt{4\sin^2 \frac{\theta}{2}}} =\frac{1}{1-2\sin \frac{\theta}{2}} = \lambda.
    $ The lemma follows.
\end{proof}

\begin{lemma}\label{lem:xy}
Let $a,b,c$ be three points in the plane. Let $\theta =|\angle bac|$ and let $\lambda > 1$ be a constant. Suppose that $\cos \theta > \frac{1}{\lambda}$, $||bc|| < ||ab||$ and $\frac{||ac||}{||ab||} = \frac{2\lambda^2\cos\theta -2\lambda}{\lambda^2-1}$.
Then $||ad||+\lambda||bd|| \leq \lambda||ab||$ for all points $d$ in the line segment $ac$.
\end{lemma}
\begin{proof}
Without loss of generality, let $||ab|| = 1$. Let $x=||ad||$.  Then $||bd|| = \sqrt{1+x^2-2x\cos\theta}$. See Figure~\ref{fig:lem2}.

\input{fig-lem2}

Note that \begin{align}
& \lambda^2-2\lambda + 1 \geq 0\nonumber\\
\Rightarrow~~~  & \lambda^2-2\lambda\cos\theta + 1 \geq 0\nonumber\\
\Rightarrow~~~  & 2\lambda\cos\theta - 2 \leq \lambda^2-1\nonumber\\
\Rightarrow~~~  & \frac{\lambda(2\lambda\cos\theta - 2)}{\lambda^2-1} \leq \lambda.\nonumber
\end{align}
Therefore $x\leq ||ac|| = \frac{2\lambda^2\cos\theta -2\lambda}{\lambda^2-1} \leq \lambda$.
Solve $||ad||+\lambda||bd|| = x+\lambda\sqrt{1+x^2-2x\cos\theta} = \lambda = \lambda||ab||$ for $x \in (0,\lambda]$, we have
\begin{align}
& x+\lambda\sqrt{1+x^2-2x\cos\theta} = \lambda\nonumber\\
\Leftrightarrow~~~  & \lambda\sqrt{1+x^2-2x\cos\theta} = \lambda-x\nonumber\\
\Leftrightarrow~~~  & \lambda^2(1+x^2-2x\cos\theta) = (\lambda-x)^2\nonumber\\
\Leftrightarrow~~~  & \lambda^2(x^2-2x\cos\theta) = x^2-2\lambda x\nonumber\\
\Leftrightarrow~~~  & \lambda^2(x-2\cos\theta) = x-2\lambda\nonumber\\
\Leftrightarrow~~~  & (\lambda^2-1)x = 2\lambda^2\cos\theta -2\lambda\nonumber\\
\Leftrightarrow~~~  & x = \frac{2\lambda^2\cos\theta -2\lambda}{\lambda^2-1} = ||ac||.\nonumber
\end{align} This implies that $||ac||+\lambda||bc|| = \lambda||ab||$.

Let $\gamma = |\angle dba|$ and $\omega = |\angle adb|$. By the law of sines in the triangle $\triangle abd$, we have
\begin{align}
\frac{||bd||}{\sin\theta} = \frac{||ad||}{\sin\gamma} = \frac{||ab||}{\sin\omega}.\label{lawofsines}
\end{align}
Therefore $$\frac{||ad||}{||ab||-||bd||} =\frac{\sin\gamma}{\sin\omega-\sin\theta}= \frac{\sin\gamma}{\sin(\pi-\theta-\gamma)-\sin\theta}= \frac{\sin\gamma}{\sin(\theta+\gamma)-\sin\theta}.$$

Define a function $$f = \frac{\sin\gamma}{\sin(\theta+\gamma)-\sin\theta}.$$
We will show
$\frac{\partial{f}}{\partial{\gamma}} \geq 0.
$ This is sufficient for the lemma because we can transform the triangle $\triangle abd$ to triangle $\triangle abc$ by moving $d$ toward $c$ (i.e., by increasing $\gamma$).

By a standard calculation,
\begin{align}
\frac{\partial{f}}{\partial{\gamma}} &= \frac{\cos\gamma(\sin(\theta+\gamma)-\sin\theta)-\sin\gamma\cos(\theta+\gamma)}{(\sin(\theta+\gamma)-\sin\theta)^2} \nonumber\\
&= \frac{\cos\gamma\sin(\theta+\gamma)-\cos\gamma\sin\theta-\sin\gamma\cos(\theta+\gamma)}{(\sin(\theta+\gamma)-\sin\theta)^2} \nonumber\\
&= \frac{\sin\theta-\cos\gamma\sin\theta}{(\sin(\theta+\gamma)-\sin\theta)^2} \nonumber\\
&= \frac{\sin\theta(1-\cos\gamma)}{(\sin(\theta+\gamma)-\sin\theta)^2}.
\end{align} We have $\frac{\partial{f}}{\partial{\gamma}} \geq 0$ because $\sin\theta > 0$, $\cos\gamma \leq 1$, and $||bc|| < ||ab||$ (and hence $\sin(\theta+\gamma)>\sin\theta$). This proves the lemma.

\end{proof}

\section{$Y_5$ is a Spanner} \label{section:main}

Let $\rho = 2+\sqrt{3} \approx 3.74$. Fix a constant $\overline{\theta}=\arccos(1-\frac{1}{\rho}) = \arccos(\sqrt{3}-1)\approx 0.75$. It is easy to verify that $$\rho  = \frac{1}{1-\cos\overline{\theta}} = \frac{1}{1-2\sin(\overline{\theta}/2)}.$$

This section contains a proof for the following main theorem.
\begin{theorem}\label{thm:main}
$Y_5$ is a $\rho$-spanner, where $\rho = 2+\sqrt{3} \approx 3.74$.
\end{theorem}

Let $G$ be a $Y_5$ graph with point set $S$. We will prove that for any pair of points $u,v \in S$, $d(u,v) \leq \rho\cdot||uv||$. We proceed by induction on the ordering $\prec$ of the pairs of points in $S$ (which is based on the Euclidean distance $||uv||$). For the base case where $\{u,v\}$ is the first pair in the ordering $\prec$, $u,v$ is connected in $G$, and hence $d(u,v) =||uv|| \leq \rho\cdot||uv||$.

For the inductive step, we will prove $d(u,v) \leq \rho\cdot||uv||$ based on the inductive hypothesis that $d(x,y) \leq \rho\cdot||xy||$ for all pairs of points $x,y\in S$ with $\{x,y\} \prec \{u,v\}$.  Without loss of generality, assume $||uv|| = 1$.

Because of Fact~\ref{fact2}, we can assume that $v$ is in the first cone of $u$, i.e., $v \in C_1^u$. Furthermore, we can assume that $v$ is on or below the bisector of $C_1^u$ because otherwise by Fact~\ref{fact2} we can mirror-flip the geometry along the bisector of $C_1^u$. Let $A_1^u(v)$ be the arc centered at $u$ with radius $||uv||$ that spans cone $C_1^u$. Let $a$ and $b$ be the start and end of the arc $A_1^u(v)$ (i.e., $a$ is the intersection of $A_1^u(v)$ and the start-ray of  $C_1^u$ and $b$ is the intersection of $A_1^u(v)$ and  the end-ray of $C_1^u$). Let $F_1^u(v)$ be the fan-shaped region enclosed by $ua$, $ub$ and $A_1^u(v)$. See Figure~\ref{fig:main} for an illustration. It is easy to verify that $u$ is in the third cone of $v$, i.e., $u \in C_3^v$. Similarly, let $A_3^v(u)$ be the arc centered at $v$ with radius $||uv||$ that spans cone $C_3^v$. Let $c$ and $d$ be the start and end of the arc $A_3^v(u)$. Let $F_3^v(u)$ be the fan-shaped region enclosed by $vc$, $vd$ and $A_3^v(u)$.

\input{fig-main}

We can assume that $u,v$ is not connected in $G$ because otherwise $d(u,v) =||uv|| \leq \rho\cdot||uv||$. Therefore, there exists a point $w \in F_1^u(v)$ such that $uw \in G$ and a point $z \in F_3^v(u)$ such that $zv \in G$. Let $$\alpha = |\angle vuw| \mbox{~~~and~~~} \beta = |\angle zvu|.$$  Let $s$ be the intersection of the rays $\overrightarrow{ub}$ and $\overrightarrow{vc}$ and let $t$ be the intersection of the rays $\overrightarrow{uw}$ and $\overrightarrow{vz}$. See Figure~\ref{fig:main} for an illustration.  It is easy to see that $|\angle usv|= 2\pi/5$ because $\overrightarrow{us}$ and $\overrightarrow{vs}$ are the boundaries of the cones $C_1^u$ and $C_3^v$ respectively. Let $\varphi=|\angle utv|$. Then
\begin{align}
\varphi=|\angle utv| = \pi - \alpha-\beta \geq \pi-|\angle vus| - |\angle svu| = |\angle usv|= 2\pi/5.\label{philower}
\end{align}

Since $\alpha+\beta = \pi-\varphi \leq \pi-2\pi/5=3\pi/5$, we have \begin{align}\min(\alpha, \beta) \leq 3\pi/10.\label{beta}
\end{align}
Based on the simple observation of (\ref{beta}), one can apply Lemma~\ref{lem:bdd} to easily prove that the stretch factor of $Y_5$ is at most $\frac{1}{1-2\sin{(3\pi/20)}}\approx 10.87$, which is the same result obtained in an earlier version of this paper~\cite{early} and, independently, in~\cite{personal}. Here we apply a more careful analysis to obtain a tighter upper bound on the stretch factor of $Y_5$.

We consider three paths between $u$ and $v$:
\begin{enumerate}
\item $P_1$ consists of the edge $(u,w) \in G$ and the shortest path from $w$ to $v$. The length of $P_1$ is $|P_1| =||uw||+d(v,w)$.
\item $P_2$ consists of the edge $(v,z) \in G$ and the shortest path from $z$ to $u$. The length of $P_2$ is $|P_2| =||vz||+d(u,z)$.
\item $P_3$ consists of the edge $(u,w) \in G$, the shortest path from $w$ to $z$, and the edge $(z,v) \in G$. The length of $P_3$ is $|P_3| =||uw||+||vz||+d(z,w)$.
\end{enumerate}
Clearly, $d(u,v) \leq \min(|P_1|, |P_2|, |P_3|)$.

Define three values
\begin{align}
g_1&=||uw||+\rho||vw||,\\
g_2&=||vz||+\rho||uz||,\\
g_3&=||uw||+||vz||+\rho||zw||.
\end{align}
In order to prove the theorem, it suffices to prove that \begin{align}
\min(g_1,g_2,g_3) \leq \rho||uv||.\label{thm1:goal}
\end{align} Here is why: if $g_1=||uw||+\rho||vw||\leq \rho||uv||$, then $||vw|| < ||uv||$ and by the inductive hypothesis $d(v,w) \leq \rho||vw||$, which gives us $$|P_1| =||uw||+d(v,w)\leq ||uw||+\rho||vw|| \leq \rho||uv||.$$ Similarly, if $g_2 \leq \rho||uv||$ then $|P_2|\leq \rho||uv||$ and if $g_3 \leq \rho||uv||$ then $|P_3|\leq \rho||uv||$. In any of the these cases, we have $d(u,v) \leq \min(|P_1|, |P_2|, |P_3|) \leq \rho||uv||$ and the theorem is proven.

In the following, we will prove (\ref{thm1:goal}) using analysis and geometric observations. We start by bounding the values of $\alpha$ and $\beta$.

If $\alpha \leq \overline{\theta}$, then by Lemma~\ref{lem:bdd},
$$||uw||+\frac{1}{1-2\sin(\overline{\theta}/2)}\cdot||vw|| \leq \frac{1}{1-2\sin(\overline{\theta}/2)}\cdot||uv||.$$ Since  $\rho = \frac{1}{1-2\sin(\overline{\theta}/2)}$, this implies
\begin{align}
g_1&=||uw||+\rho||vw|| \leq \rho||uv||,
\end{align} and we are done. Similarly, if $\beta \leq \overline{\theta}$, then $g_2 = ||vz||+\rho||uz||\leq \rho||uv||$ and we are done.

Therefore we can assume $\alpha > \overline{\theta}$ and $\beta > \overline{\theta}$. Since $v$ is on or below the bisector of $C_1^u$, we have $|\angle auv| \leq \pi/5 < \overline{\theta}$ and $|\angle uvd| \leq \pi/5 < \overline{\theta}$. This implies that neither $z$ or $w$ is below the line $uv$. So we can assume that both $z$ and $w$ are above the line $uv$, as illustrated by Figure~\ref{fig:main}.

The following proposition plays a key role in this proof.
\begin{proposition}\label{prop:1}
If $g_1 > \rho||uv||$ and $g_2> \rho||uv||$, then $||wz|| \leq 2\cos\overline{\theta}-1.$
\end{proposition}
\begin{proof}
Let $w', w''$ be two points in the ray $\overrightarrow{uw}$ such that $$||uw'||=\frac{2\rho^2\cos\alpha -2\rho}{\rho^2-1} \mbox{~~~and~~~} ||uw''| = 1.$$ By Lemma~\ref{lem:xy}, if $||uw|| \leq ||uw'||$ then $g_1 = ||uw||+\rho||vw|| \leq \rho||uv||$. So we can assume $w$ is in the line segment $w'w''$. See Figure~\ref{fig:frac-1}.

\input{fig-frac-1}

Similarly, let $z', z''$ be two points in the ray $\overrightarrow{vz}$ such that $$||vz'||=\frac{2\rho^2\cos\beta -2\rho}{\rho^2-1} \mbox{~~~and~~~} ||vz''| = 1.$$ Since $g_2> \rho||uv||$, we can assume $z$ is in the line segment $z'z''$.

By linearity, we have $$||wz|| \leq \max(||w'z'||,||w'z''||,||w''z'||,||w''z''||).$$
By the law of sines in the triangle $\triangle uvt$, we have (recall that we assume $||uv||$ = 1):
$$||ut|| = \frac{\sin\beta}{\sin(\alpha+\beta)}  \mbox{~~~and~~~} ||vt|| = \frac{\sin\alpha}{\sin(\alpha+\beta)}.$$ See Figure~\ref{fig:frac-1} for illustration.

We continue by distinguishing two cases.

\vspace*{0.5cm}
\noindent{\bf Case 1.} First consider the case where $uw$ and $vs$ cross each other. See Figure~\ref{fig:propcase1} (a). In this case, since $wz$ is a line segment in the triangle $\triangle tw''z''$, we have $||wz|| \leq \max(||tw''||, ||tz''||, ||w''z''||)$. Since $\alpha \geq  \overline{\theta}$, $\beta \geq  \overline{\theta}$ and $\sin(\alpha+\beta) \leq 1$, we have $||tw''|| = ||uw''| - ||ut|| = 1-\frac{\sin\beta}{\sin(\alpha+\beta)} \leq 1-\sin\overline{\theta} < 2\cos\overline{\theta}-1$ and $||tz''|| = 1-\frac{\sin\alpha}{\sin(\alpha+\beta)} \leq 1-\sin\overline{\theta} < 2\cos\overline{\theta}-1.$ Now consider $||w''z''||$. It is easy to see that $||w''z''||$ increases when we fix $vz''$ and rotate $uw''$ clockwise around $u$ until $\alpha = \overline{\theta}$. Similarly, $||w''z''||$ increases when we fix $uw''$ and rotate $vz''$ counterclockwise around $v$ until $\beta = \overline{\theta}$. Therefore $||w''z''||$ is maximized when $\alpha=\beta=\overline{\theta}$. See Figure~\ref{fig:propcase1} (b). In this case it is a simple calculation based on the geometry to verify that $||w''z''|| = 2\cos\overline{\theta}-1$.

\input{fig-propcase1}

\input{fig-propcase2}

\vspace*{0.5cm}
\noindent{\bf Case 2.} Now assume that $uw$ and $vs$ do not cross each other. See Figure~\ref{fig:propcase2} (a). In this case either $||uw|| < ||ut||$ or $||vz|| < ||vt||$ or both. If $||vz|| < ||vt||$, then $||wz||$ increases when we fix $vz$ and rotate $\overrightarrow{uw}$ counterclockwise around $u$ until $\alpha = 3\pi/5-\beta$. Otherwise we have $||uw|| < ||ut||$; then $||wz||$ increases when we fix $uw$ and rotate $\overrightarrow{vz}$ clockwise around $v$ until $\beta = 3\pi/5-\alpha$. Note that in the above rotating process, it is possible for $\overrightarrow{uw}$ or $\overrightarrow{vz}$ to go beyond the boundaries of the cones $C_1^u$ or $C_3^v$ respectively, but this is not a problem because we only need to bound $||wz||$ in this proposition and going beyond the boundaries of the cones does not affect the discussion that follows. So in either case, we can assume $\alpha+\beta = 3\pi/5$.

Since $||uw'|| = \frac{2\rho^2\cos\alpha -2\rho}{\rho^2-1}$ decreases when  $\alpha$ increases, $w$ is still in the line segment $w'w''$ after rotation. Similarly, $z$ is in the line segment $z'z''$ after rotation. This means that $$||wz|| \leq \max(||w'z'||,||w'z''||,||w''z'||,||w''z''||)$$ still holds after the rotation. See Figure~\ref{fig:propcase2}~(b). Without loss of generally, assume that $\alpha \geq \beta$. Therefore $3\pi/10 \leq \alpha \leq 3\pi/5-\overline{\theta}$ and $\overline{\theta} \leq \beta \leq 3\pi/10$. Let $c_1 = \frac{2\rho^2}{\rho^2-1}$ and $c_2=\frac{1}{\sin(3\pi/5)}$. We have
    \begin{align}
\frac{d ||uw'||}{d \alpha} &= \frac{d (\frac{2\rho^2\cos\alpha -2\rho}{\rho^2-1})}{d \alpha} = \frac{-2\rho^2\sin\alpha}{\rho^2-1} = -c_1\sin\alpha,\label{firstd}\\
\frac{d ||vz'||}{d \alpha} &= \frac{d (\frac{2\rho^2\cos\beta -2\rho}{\rho^2-1})}{d \alpha} = \frac{d (\frac{2\rho^2\cos(3\pi/5-\alpha) -2\rho}{\rho^2-1})}{d \alpha} =\frac{2\rho^2\sin(3\pi/5-\alpha)}{\rho^2-1} = c_1\sin(3\pi/5-\alpha),\\
\frac{d ||ut||}{d \alpha} &= \frac{d (\frac{\sin\beta}{\sin(\alpha+\beta)})}{d \alpha} = \frac{d (\frac{\sin(3\pi/5-\alpha)}{\sin(3\pi/5)})}{d \alpha} =  \frac{-\cos(3\pi/5-\alpha)}{\sin(3\pi/5)} = -c_2\cos(3\pi/5-\alpha),\\
\frac{d ||vt||}{d \alpha} &= \frac{d (\frac{\sin\alpha}{\sin(\alpha+\beta)})}{d \alpha} = \frac{d (\frac{\sin\alpha}{\sin(3\pi/5)})}{d \alpha} = \frac{\cos\alpha}{\sin(3\pi/5)} = c_2\cos\alpha.\label{lastd}
\end{align}
Let
\begin{align}
x_1 &= ||ut|| - ||uw'|| = \frac{\sin\beta}{\sin(\alpha+\beta)} - \frac{2\rho^2\cos\alpha -2\rho}{\rho^2-1} = \frac{\sin(3\pi/5-\alpha)}{\sin(3\pi/5)} - \frac{2\rho^2\cos\alpha -2\rho}{\rho^2-1},\label{x1}\\
x_2 &= ||uw''|| - ||ut|| = 1- \frac{\sin\beta}{\sin(\alpha+\beta)} = 1- \frac{\sin(3\pi/5-\alpha)}{\sin(3\pi/5)},\label{x2}\\
y_1 &= ||vt|| - ||vz'|| = \frac{\sin\alpha}{\sin(\alpha+\beta)} - \frac{2\rho^2\cos\beta -2\rho}{\rho^2-1} = \frac{\sin\alpha}{\sin(3\pi/5)} - \frac{2\rho^2\cos\beta -2\rho}{\rho^2-1},\label{y1}\\
y_2 &= ||vz''|| - ||vt|| = 1 - \frac{\sin\alpha}{\sin(\alpha+\beta)} = 1- \frac{\sin\alpha}{\sin(3\pi/5)}\label{y2}.
\end{align}
Note that the values of $x_1$ and $y_1$ can be positive or negative. From (\ref{firstd}) - (\ref{lastd}), we have
\begin{align}
\frac{d x_1}{d \alpha} &= \frac{d(||ut|| - ||uw'||)}{d \alpha} = -c_2\cos(3\pi/5-\alpha) + c_1\sin\alpha,\label{d1}\\
\frac{d x_2}{d \alpha} &= \frac{d(||uw''|| - ||ut||)}{d \alpha} = \frac{d(1 - ||ut||)}{d \alpha} = c_2\cos(3\pi/5-\alpha),\\
\frac{d y_1}{d \alpha} &= \frac{d(||vt|| - ||vz'||)}{d \alpha} = c_2\cos\alpha - c_1\sin(3\pi/5-\alpha),\\
\frac{d y_2}{d \alpha} &= \frac{d(||vz''|| - ||vt||)}{d \alpha} = \frac{d(1 - ||vt||)}{d \alpha} = -c_2\cos\alpha.\label{d4}
\end{align}
Recall that $c_1 = \frac{2\rho^2}{\rho^2-1}$, $c_2=\frac{1}{\sin(3\pi/5)}$, and $3\pi/10 \leq \alpha \leq 3\pi/5-\overline{\theta}$, we verify the following:
\begin{align}
\frac{d^2 x_1}{d \alpha^2} &= -c_2\sin(3\pi/5-\alpha)+c_1\cos\alpha > - 1.1 \cdot\sin(3\pi/10)+2.1 \cdot\cos(3\pi/5-\overline{\theta})  > 0,\\
\frac{d^2 x_2}{d \alpha^2} &= c_2\sin(3\pi/5-\alpha) > 0,\\
\frac{d^2 y_1}{d \alpha^2} &= -c_2\sin\alpha+c_1\cos(3\pi/5-\alpha) > - 1.1 \cdot\sin(3\pi/5-\overline{\theta})+2.1 \cdot\cos(3\pi/10)  > 0,\\
\frac{d^2 y_2}{d \alpha^2} &= c_2\sin\alpha > 0.\\
\end{align}
Therefore, by plugging $\alpha = 3\pi/10$ or $\alpha = 3\pi/5-\overline{\theta}$ as the lower- or upper-bound of $\alpha$ into (\ref{d1})-(\ref{d4}), we can verify the following ranges:
\begin{align}
-c_2\cos(3\pi/10) + c_1\sin(3\pi/10) \leq & \frac{d x_1}{d \alpha} \leq -c_2\cos\overline{\theta} + c_1\sin(3\pi/5-\overline{\theta}),\\
c_2\cos(3\pi/10) \leq & \frac{d x_2}{d \alpha} \leq c_2\cos\overline{\theta},\label{dx2}\\
c_2\cos(3\pi/10) - c_1\sin(3\pi/10) \leq & \frac{d y_1}{d \alpha} \leq c_2\cos(3\pi/5-\overline{\theta})-c_1\sin\overline{\theta},\\
-c_2\cos(3\pi/10) \leq & \frac{d y_2}{d \alpha} \leq -c_2\cos(3\pi/5-\overline{\theta}).
\end{align}

Specifically, we can verify that
\begin{align}
\frac{d x_1}{d \alpha} \geq \max(\frac{d x_2}{d \alpha},\left|\frac{d y_1}{d \alpha}\right|,\left|\frac{d y_2}{d \alpha}\right|),\label{compared}
\end{align}
which implies
$\frac{d (x_1-x_2)}{d \alpha} = \frac{d x_1}{d \alpha}-\frac{d x_2}{d \alpha} > 0$. By simply plugging $\alpha = 3\pi/10$ into (\ref{x1}) and (\ref{x2}), we verify that $(x_1-x_2) > 0$ when $\alpha = 3\pi/10$ and hence $x_1 > x_2$ for all $\alpha \in [3\pi/10, 3\pi/5-\overline{\theta}]$. Similarly, we have $x_2 > 0$ when $\alpha = 3\pi/10$, and hence by (\ref{dx2}), $x_2 > 0$ for all $\alpha \in [3\pi/10, 3\pi/5-\overline{\theta}]$. Now we have $x_1 > x_2 > 0$.

By the triangle inequality,
\begin{align}
||w'z'|| &\leq ||tw'||+||tz'|| = |x_1| + |y_1| = x_1 + |y_1|,\\
||w'z''|| &\leq ||tw'||+||tz''|| = |x_1| + |y_2| = x_1 + |y_2|,\\
||w''z'|| &\leq ||tw''||+||tz'|| = |x_2| + |y_1| = x_2 + |y_1|\leq x_1 + |y_1|,\\
||w''z''|| &\leq ||tw''||+||tz''|| = |x_2| + |y_2| = x_2 + |y_2|\leq x_1 + |y_2|.
\end{align}
By (\ref{compared}),
\begin{align}
\frac{d (x_1 + |y_1|)}{d \alpha} & \geq \frac{d (x_1)}{d \alpha} - \left|\frac{d (y_1)}{d \alpha}\right| \geq 0,\\
\frac{d (x_1 + |y_2|)}{d \alpha} & \geq \frac{d (x_1)}{d \alpha} - \left|\frac{d (y_2)}{d \alpha}\right| \geq 0.
\end{align}
By plugging $\alpha = 3\pi/5-\overline{\theta}$ into (\ref{x1}), (\ref{y1}), and (\ref{y2}), one can easily verify that $x_1 + |y_1| \leq 2\cos\overline{\theta}-1$ and $x_1 + |y_2| \leq 2\cos\overline{\theta}-1$ when $\alpha = 3\pi/5-\overline{\theta}$ (i.e., when $\alpha$ is maximized). Therefore  $\max(x_1 + |y_1|, x_1 + |y_2|) \leq 2\cos\overline{\theta}-1$ for all $\alpha \in [3\pi/10, 3\pi/5-\overline{\theta}]$, and hence $||wz|| \leq \max(||w'z'||,||w'z''||,||w''z'||,||w''z''||) \leq 2\cos\overline{\theta}-1 $ as required.

This proves that $||wz|| \leq 2\cos\overline{\theta}-1.$

\end{proof}

The theorem follows immediately from Proposition~\ref{prop:1}: If $g_1 \leq \rho||uv||$ or $g_2 \leq \rho||uv||$, then we are done; otherwise by Proposition~\ref{prop:1} $$g_3 =||uw||+||vz||+\rho||zw|| \leq 1+1+\rho(2\cos\overline{\theta}-1) = \rho,$$ since $\cos\overline{\theta} = 1-\frac{1}{\rho}$. Therefore we have $\min(g_1,g_2,g_3) \leq \rho$, as required. This completes the proof of the main theorem.

\section{A lower bound on the stretch factor of $Y_5$}\label{section:lower_bound}

\input{fig-frac-4}

The preceding inductive proof of the upper bound on the stretch factor of $Y_5$ suggests a possible construction that gives a lower bound of the stretch factor of $Y_5$. It is based on recursively attaching the ``lattice'' as shown in Figure~\ref{fig:propcase1} (b) to pairs of non-adjacent points (e.g., pairs $\{u,z\}, \{z,w\}, \{w,v\}$ in Figure~\ref{fig:propcase1} (b)). This recursion-based construction results in a ``fractal'' starting from the pair $\{u,v\}$. See Figure~\ref{fig:frac-4} (a). However, the growth of fractal is limited because neighboring fractal branches collide into each other, thereby creating shortcuts to the paths, as shown in the circled area of Figure~\ref{fig:frac-4} (a). This lowers the stretch factor of the fractal to 2.66.

We adjust the shape of the fractal to increase the stretch factor. In Figure~\ref{fig:frac-4} (b), we obtained a stretch factor of more than 2.87 by equalizing the length of all shortest paths between $u$ and $v$, as shown in Figure~\ref{fig:frac-4} (b). The exact locations of the points are given in the appendix.

\section{$YY_5$ is not a Spanner} \label{section:yy}
We give a construction of a $YY_5$ graph whose stretch factor is unbounded. Figure~\ref{fig:yy} shows the initial steps of constructing such a $YY_5$ graph, where the path between $a$ and $b$ can grow horizontally to the right by adding more points following the pattern, exceeding any bound on the stretch factor.

\input{fig-yy}

\section{Concluding Remarks}
\label{sec:conclusion}
In this paper we prove that the stretch factor of $Y_5$ is in the interval (2.87, 3.74). While the gap between the upper bound of 3.74 and the lower bound of 2.87 proved in this paper is small, the tight bound of the stretch factor of $Y_5$ remains unknown. Similarly, it will be interesting to study the tight bounds of other Yao graphs $Y_k$ for $k \geq 4$.

Clearly, the Yao-Yao graphs are less well understood than the Yao graphs. While we know some partial results on the stretch factors of Yao-Yao graphs, many questions about the spanning properties of Yao-Yao graphs remain unresolved. For example, are the Yao-Yao graphs spanners for {\em all} $k > 6$?

\bibliographystyle{plain}
\bibliography{ref}

\section{Appendix}
The following are the locations of the points in the $Y_5$ graph shown in Figure~\ref{fig:frac-4} (b) whose stretch factor is more than 2.87.

\begin{tabular}{l l}
\hline
\multicolumn{2}{c}{Locations of the points in Figure~\ref{fig:frac-4} (b)} \\
\hline
(0, 0)& $u$\\
(252, 82)& $v$\\
(130, 230)& $z$\\
(12, 193)& $w$\\
(30, 302)& \\
(293, 269)& \\
(321, 229)& \\
(-143, 130)& \\
(-143, 80)& \\
(193, 384)& \\
(158, 367)& \\
(-135, 272)& \\
(-91, 287)& \\
(-153, -55)& \\
(371, 75)& \\
(410, 115)& \\
(334, 276)& \\
(341, 264)& \\
(-179, 97)& \\
(-180, 112)& \\
(-91, -75)& \\
(316, 36)& \\
(352, 229)& \\
(303, 297)& \\
(-167, 63)& \\
(-167, 147)& \\
(-26, -75)& \\
(371, 213)& \\
(51, 310)& \\
(-176, 37)& \\
(344, 274)& \\
(-189, 105)& \\
(99, 320)& \\
(-15, 284)& \\
\hline
\end{tabular}

\end{document}

%% file: fig-lem2.tex

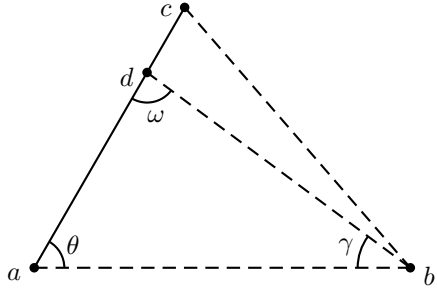
\begin{figure}[tbhp]
\begin{center}
\begin{pspicture}(6,6)

	\pnode(0,0){a}\uput[200](a){$a$}
	\pnode(5,0){b}\uput[-20](b){$b$}
	\pnode([nodesep=4,,angle=60]{b}a){c}\uput[190](c){$c$}
	
	\pnode([nodesep=3]{c}a){d}\uput[190](d){$d$}

	\psline(a)(c)
	\psline[linestyle=dashed](a)(b)
	\psline[linestyle=dashed](c)(b)
\psline[linestyle=dashed](b)(d)
	
	\pstMarkAngle[MarkAngleRadius=0.4,  LabelSep=0.6]{b}{a}{d}{$\theta$}
	\pstMarkAngle[MarkAngleRadius=0.7,  LabelSep=0.9]{d}{b}{a}{$\gamma$}
	\pstMarkAngle[MarkAngleRadius=0.4,  LabelSep=0.6]{a}{d}{b}{$\omega$}
	
\psdot(a
)
\psdot(b
)
\psdot(c
)
\psdot(d
)

\end{pspicture}
\caption{Illustration for the proof of Lemma~\ref{lem:xy}}\label{fig:lem2}
\end{center}
\end{figure}


%% file: fig-main.tex
\begin{figure}[tbhp]
\begin{center}
\begin{pspicture}(8,9)
  \pnode(2,3){u}\uput[180](u){$u$}

\pswedge[fillcolor=lightgray,linewidth=0pt,linecolor=lightgray,fillstyle=solid,origin={u}](u){5}{0}{72}

  \pnode(8,3){x}
  \pnode(7,3){a}\uput[-90](a){$a$}

  \pnode([nodesep=5,angle=12]{x}u){v}\uput[0](v){$v$}
\pswedge[fillcolor=lightgray,linewidth=0pt,linecolor=lightgray,fillstyle=solid,origin={v}](v){5}{144}{216}

  \pnode([nodesep=6,angle=72]{x}u){y}
  \pnode([nodesep=5,angle=72]{x}u){b}\uput[150](b){$b$}

  \pnode([nodesep=4.4,angle=50]{v}u){w}\uput[20](w){$w$}

  \pnode([nodesep=6,angle=36]{x}u){z}

  \psline[linewidth=0.2pt]{->}(u)(x)
  \psline[linewidth=0.2pt]{->}(u)(y)
  \psline[linestyle=dotted]{->}(u)(z)
  \pnode([nodesep=5.3,angle=50]{x}u){zz}\uput[90](zz){$C_1^u$}


  \psline[linestyle=dashed](u)(v)

\pnode([nodesep=6,angle=-12]{u}v){x1}
\pnode([nodesep=6,angle=36]{x1}v){yy}
\pnode([nodesep=6,angle=-36]{x1}v){xx}
  \psline[linewidth=0.2pt]{->}(v)(yy)
  \psline[linewidth=0.2pt]{->}(v)(xx)

  \psline[linestyle=dotted]{->}(v)(x1)

\pnode([nodesep=5.3,angle=-26]{u}v){xx1}\uput[90](xx1){$C_3^v$}

  \pnode([nodesep=5,angle=36]{x1}v){d}\uput[-60](d){$d$}
  \pnode([nodesep=5,angle=-36]{x1}v){c}\uput[90](c){$c$}

  \pnode([nodesep=4.6]{c}v){s}\uput[90](s){$s$}


  \pnode([nodesep=4.8,angle=-43]{u}v){z}\uput[180](z){$z$}

  \pnode([nodesep=3.84]{z}v){t}\uput[120](t){$t$}

  \psline[linecolor=red](u)(w)

  \psline[linecolor=red](v)(z)

  \psline[linecolor=red,linestyle=dashed](u)(z)
  \psline[linecolor=red,linestyle=dashed](v)(w)
  \psline[linecolor=red,linestyle=dashed](z)(w)

  \pstMarkAngle[MarkAngleRadius=0.4, linewidth=0.5pt, LabelSep=0.6]{u}{t}{v}{$\varphi$}
  \pstMarkAngle[MarkAngleRadius=0.6, linewidth=0.5pt, LabelSep=0.8]{v}{u}{w}{$\alpha$}
  \pstMarkAngle[MarkAngleRadius=0.6, linewidth=0.5pt, LabelSep=0.8]{z}{v}{u}{$\beta$}

  \psdot(u)
\psdot(v)
\psdot(z)
\psdot(w)

\end{pspicture}
\caption{Illustration for the proof of Theorem~\ref{thm:main}.}\label{fig:main}
\end{center}
\end{figure}

%% file: fig-frac-1.tex
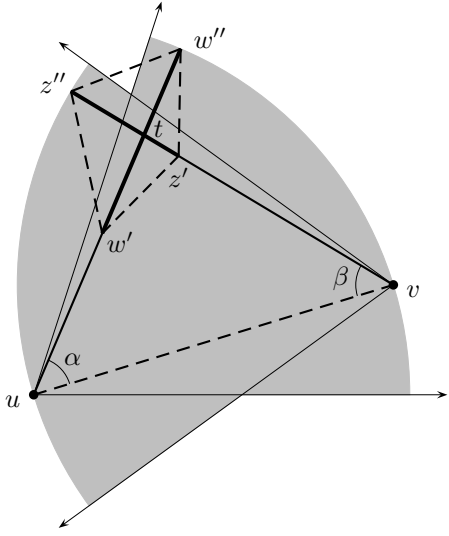
\begin{figure}[tbhp]
\begin{center}
\begin{pspicture}(-2,-2)(6,6)

\pnode(0,0){O}
\pnode(4.78153, 1.46186){v'}

\pswedge[fillcolor=lightgray,linewidth=0pt,linecolor=lightgray,fillstyle=solid,origin={O}](O){5}{0}{72}
\pswedge[fillcolor=lightgray,linewidth=0pt,linecolor=lightgray,fillstyle=solid,origin={v'}](v'){5}{144}{216}
	
	\pnode([nodesep=5.5]{O}O){xu}
	\psline[linewidth=0.2pt,arrowsize=0.1]{->}(O)(xu)
	\pnode([nodesep=5.5,angle=72]{O}O){yu}
	\psline[linewidth=0.2pt,arrowsize=0.1]{->}(O)(yu)
	
	\pnode([nodesep=5.5,angle=144]{v'}v'){xv'}
	\psline[linewidth=0.2pt,arrowsize=0.1]{->}(v')(xv')
	\pnode([nodesep=5.5,angle=216]{v'}v'){yv'}
	\psline[linewidth=0.2pt,arrowsize=0.1]{->}(v')(yv')

\rput[l]{17}(0,0){

	\pnode(0,0){u}
	\pnode(5,0){v}
	\pnode([nodesep=5,angle=50]{v}u){w''}
	\pnode([nodesep=5,angle=-48]{u}v){z''}
	\pnode([nodesep=2.333,angle=0]{w''}u){w'}
	\pnode([nodesep=3.33,angle=0]{z''}v){z'}
	
	\pnode([nodesep=3.75,angle=0]{w''}u){t}
	
	\psline[linestyle=dashed](u)(v)
	\psline(u)(w'')
	\psline(v)(z'')
	\psline[linestyle=dashed](z')(w')
	\psline[linestyle=dashed](w')(z'')
	\psline[linestyle=dashed](z'')(w'')
	\psline[linestyle=dashed](w'')(z')
	
	\psline[linewidth=1.5pt](w'')(w')
	\psline[linewidth=1.5pt](z'')(z')

}

\psdot(u
)
\psdot(v
)

	\uput[200](u){$u$}
	\uput[-20](v){$v$}
	\uput[25](w''){$w''$}
	\uput{0.06}[150](z''){$z''$}
	\uput{0.06}[-30](w'){$w'$}
	\uput{0.15}[-90](z'){$z'$}
	\uput{0.13}[20](t){$t$}
	
	\pstMarkAngle[MarkAngleRadius=0.5, linewidth=0.2pt, LabelSep=0.7]{v}{u}{w''}{$\alpha$}
	\pstMarkAngle[MarkAngleRadius=0.5, linewidth=0.2pt, LabelSep=0.7]{z''}{v}{u}{$\beta$}


\end{pspicture}
\caption{Illustration for the proof of Proposition~\ref{prop:1}. $w$ and $z$ are in the line segments $w'w''$ and $z'z''$, respectively.}\label{fig:frac-1}
\end{center}
\end{figure}

%% file: fig-propcase1.tex

\begin{figure}[tbhp]
\begin{center}
\begin{pspicture}(-2,-2)(6,6)

\pnode(0,0){O}
\pnode(4.78153, 1.46186){v'}

\pswedge[fillcolor=lightgray,linewidth=0pt,linecolor=lightgray,fillstyle=solid,origin={O}](O){5}{0}{72}
\pswedge[fillcolor=lightgray,linewidth=0pt,linecolor=lightgray,fillstyle=solid,origin={v'}](v'){5}{144}{216}
	
	\pnode([nodesep=5.5]{O}O){xu}
	\psline[linewidth=0.2pt,arrowsize=0.1]{->}(O)(xu)
	\pnode([nodesep=5.5,angle=72]{O}O){yu}
	\psline[linewidth=0.2pt,arrowsize=0.1]{->}(O)(yu)
	
	\pnode([nodesep=5.5,angle=144]{v'}v'){xv'}
	\psline[linewidth=0.2pt,arrowsize=0.1]{->}(v')(xv')
	\pnode([nodesep=5.5,angle=216]{v'}v'){yv'}
	\psline[linewidth=0.2pt,arrowsize=0.1]{->}(v')(yv')

\rput[l]{17}(0,0){

	\pnode(0,0){u}
	\pnode(5,0){v}
	\pnode([nodesep=5,angle=50]{v}u){w''}
	\pnode([nodesep=5,angle=-48]{u}v){z''}
	\pnode([nodesep=4.67,angle=0]{z''}v){z}
	\pnode([nodesep=4.29,angle=0]{w''}u){w}

	\pnode([nodesep=3.75,angle=0]{w''}u){t}
	
	\pnode([nodesep=5.5,angle=43]{v}u){w'''}
	\pnode([nodesep=5.5,angle=-43]{u}v){z'''}
	\psline[linewidth=0.2pt](u)(w''')
	\psline[linewidth=0.2pt](v)(z''')
	\pstMarkAngle[arrows=<-, arrowsize=0.15, MarkAngleRadius=5.13, linewidth=0.2pt]{w'''}{u}{w''}{}
	\pstMarkAngle[arrows=->, arrowsize=0.15, MarkAngleRadius=5.13, linewidth=0.2pt]{z''}{v}{z'''}{}
	
	\psline[linestyle=dashed](u)(v)
	\psline(u)(w'')
	\psline(v)(z'')
	\psline[linestyle=dashed](z'')(w'')
	\psline[linestyle=dashed](z)(w)
	\psline[linewidth=1.5pt](u)(w)
	\psline[linewidth=1.5pt](v)(z)

}
	
\psdot(u
)
\psdot(v
)
\psdot(w
)
\psdot(z
)

	\uput[200](u){$u$}
	\uput[-20](v){$v$}
	\uput{0.13}[90](w''){$w''$}
	\uput{0.13}[90](z''){$z''$}
	\uput{0.1}[-20](w){$w$}
	\uput{0.1}[205](z){$z$}
	\uput{0.15}[-70](t){$t$}

	\pstMarkAngle[MarkAngleRadius=0.5, linewidth=0.2pt, LabelSep=0.7]{v}{u}{w'''}{$\overline{\theta}$}
	\pstMarkAngle[MarkAngleRadius=0.5, linewidth=0.2pt, LabelSep=0.7]{z'''}{v}{u}{$\overline{\theta}$}
	\pstMarkAngle[MarkAngleRadius=1.0, linewidth=0.2pt, LabelSep=1.2]{v}{u}{w''}{$\alpha$}
	\pstMarkAngle[MarkAngleRadius=1.0, linewidth=0.2pt, LabelSep=1.2]{z''}{v}{u}{$\beta$}

\uput[-70](2.5,-1.5){(a)}

\end{pspicture}
\begin{pspicture}(-2,-2)(5,5)

\rput{18}(0,0){

	\pnode(0,0){u}
	\pnode(5,0){v}
	\pnode([nodesep=5,angle=43]{v}u){w}
	\pnode([nodesep=5,angle=-43]{u}v){z}
	
}
	
	\psline[linestyle=dashed](u)(v)
	\psline(u)(w)
	\psline(v)(z)
	\psline[linestyle=dashed](u)(z)(w)(v)
	
	\psline[linewidth=1.5pt](u)(w)
	\psline[linewidth=1.5pt](v)(z)


\psdot(u
)
\psdot(v
)

	\uput[200](u){$u$}
	\uput[-20](v){$v$}
	\uput[20](w){$w''$}
	\uput[160](z){$z''$}
	
	\pstMarkAngle[MarkAngleRadius=0.5, linewidth=0.5pt, LabelSep=0.8]{v}{u}{w}{$\overline{\theta}$}
	\pstMarkAngle[MarkAngleRadius=0.5, linewidth=0.5pt, LabelSep=0.8]{z}{v}{u}{$\overline{\theta}$}

\uput[-70](2.5,-1.5){(b)}

\end{pspicture}
\caption{Illustration for the case 1 of the proof of Proposition~\ref{prop:1}. (a) illustrates the rotation. (b) shows that $||w''z''||$ is maximized when $\alpha=\beta=\overline{\theta}$.}\label{fig:propcase1}
\end{center}
\end{figure}
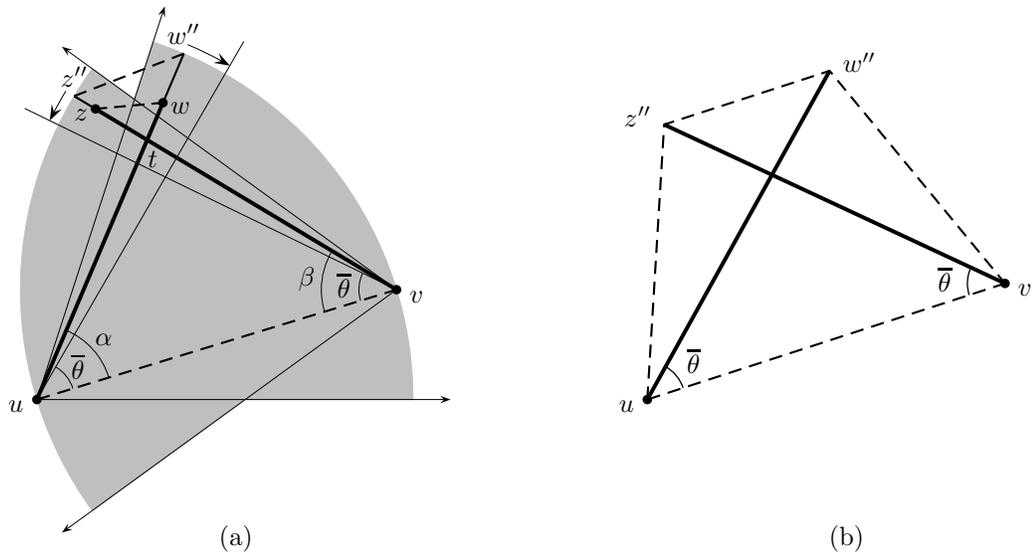


%% file: fig-propcase2.tex

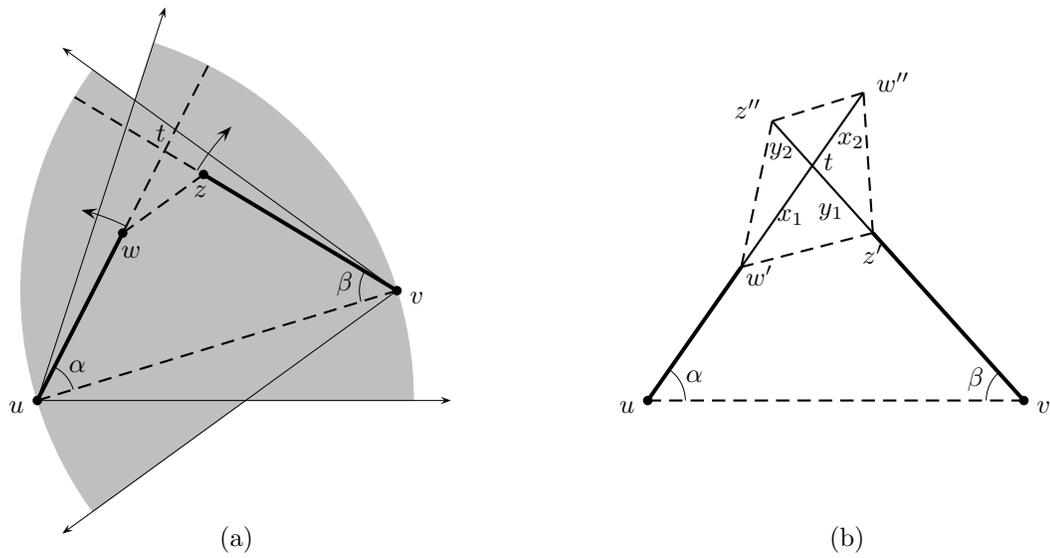
\begin{figure}[tbhp]
\begin{center}
\begin{pspicture}(-2,-2)(6,6)

\pnode(0,0){O}
\pnode(4.78153, 1.46186){v'}

\pswedge[fillcolor=lightgray,linewidth=0pt,linecolor=lightgray,fillstyle=solid,origin={O}](O){5}{0}{72}
\pswedge[fillcolor=lightgray,linewidth=0pt,linecolor=lightgray,fillstyle=solid,origin={v'}](v'){5}{144}{216}
	
	\pnode([nodesep=5.5]{O}O){xu}
	\psline[linewidth=0.2pt,arrowsize=0.1]{->}(O)(xu)
	\pnode([nodesep=5.5,angle=72]{O}O){yu}
	\psline[linewidth=0.2pt,arrowsize=0.1]{->}(O)(yu)
	
	\pnode([nodesep=5.5,angle=144]{v'}v'){xv'}
	\psline[linewidth=0.2pt,arrowsize=0.1]{->}(v')(xv')
	\pnode([nodesep=5.5,angle=216]{v'}v'){yv'}
	\psline[linewidth=0.2pt,arrowsize=0.1]{->}(v')(yv')

\rput[l]{17}(0,0){

	\pnode(0,0){u}
	\pnode(5,0){v}
	\pnode([nodesep=5,angle=46]{v}u){w''}
	\pnode([nodesep=5,angle=-48]{u}v){z''}
	
	\pnode([nodesep=2.5,angle=0]{w''}u){w}
	\pnode([nodesep=3.0,angle=0]{z''}v){z}

	\pnode([nodesep=3.7,angle=0]{w''}u){t}

\pnode([nodesep=5.5,angle=60]{O}O){yuu}	\pstMarkAngle[arrows=->, arrowsize=0.15, MarkAngleRadius=2.6, linewidth=0.2pt]{w}{u}{yuu}{}
	
\pnode([nodesep=5.5,angle=65]{O}O){yvv}	\pstMarkAngle[arrows=<-, arrowsize=0.15, MarkAngleRadius=3.1, linewidth=0.2pt]{yvv}{v}{z}{}

	\psline[linestyle=dashed](u)(v)
	\psline[linestyle=dashed](w)(z)
	\psline[linestyle=dashed](u)(w'')
	\psline[linestyle=dashed](v)(z'')	
	
	\psline[linewidth=1.5pt](u)(w)
	\psline[linewidth=1.5pt](v)(z)

}


\psdot(u
)
\psdot(v
)
\psdot(z
)
\psdot(w
)

	\uput[200](u){$u$}
	\uput[-20](v){$v$}
	\uput[-70](w){$w$}
	\uput{0.15}[-100](z){$z$}
	
	\uput{0.15}[100](t){$t$}
	
	\pstMarkAngle[MarkAngleRadius=0.5, linewidth=0.2pt, LabelSep=0.7]{v}{u}{w''}{$\alpha$}
	\pstMarkAngle[MarkAngleRadius=0.5, linewidth=0.2pt, LabelSep=0.7]{z''}{v}{u}{$\beta$}

\uput[-70](2.5,-1.5){(a)}

\end{pspicture}
\begin{pspicture}(-2,-2)(6,6)


	
	


	\pnode(0,0){u}
	\pnode(5,0){v}
	\pnode([nodesep=5,angle=55]{v}u){w''}
	\pnode([nodesep=5,angle=-48]{u}v){z''}
	\pnode([nodesep=2.167,angle=0]{w''}u){w'}
	\pnode([nodesep=3.0,angle=0]{z''}v){z'}
	
	\pnode([nodesep=3.82,angle=0]{w''}u){t}
	
	
	\psline[linestyle=dashed](u)(v)
	\psline(u)(w'')
	\psline(v)(z'')
	\psline[linestyle=dashed](z')(w')
	\psline[linestyle=dashed](w')(z'')
	\psline[linestyle=dashed](z'')(w'')
	\psline[linestyle=dashed](w'')(z')
	
	\psline[linewidth=1.5pt](u)(w')
	\psline[linewidth=1.5pt](v)(z')


\psdot(u
)
\psdot(v
)

	\uput[200](u){$u$}
	\uput[-20](v){$v$}
	\uput[25](w''){$w''$}
	\uput{0.15}[150](z''){$z''$}
	\uput{0.06}[-30](w'){$w'$}
	\uput{0.15}[-90](z'){$z'$}
	\uput{0.15}[10](t){$t$}
	
	\uput{0.6}[-111](t){$x_1$}
	\uput{0.4}[37](t){$x_2$}
	\uput{0.52}[-70](t){$y_1$}
	\uput{0.26}[152](t){$y_2$}
	
	\pstMarkAngle[MarkAngleRadius=0.5, linewidth=0.2pt, LabelSep=0.7]{v}{u}{w''}{$\alpha$}
	\pstMarkAngle[MarkAngleRadius=0.5, linewidth=0.2pt, LabelSep=0.7]{z''}{v}{u}{$\beta$}

\uput[-70](2.5,-1.5){(b)}

\end{pspicture}
\caption{Illustration for case 2 of Proposition~\ref{prop:1}}\label{fig:propcase2}
\end{center}
\end{figure}


%% file: fig-frac-4.tex

\begin{figure}[tbp]
\begin{center}
\begin{pspicture*}(0,0)(10,8)

\scalebox{1.5}{

\rput{216}(4.5,8){

	
\pnode(3.37, 5.13){p1}
\pnode(5.89, 4.31){p2}
\pnode(4.65, 2.81){p3}
\pnode(3.48, 3.21){p4}
\pnode(3.52, 2.31){p5}
\pnode(4.06, 2.12){p6}
\pnode(5.99, 2.37){p7}
\pnode(6.57, 3.10){p8}
\pnode(2.12, 3.67){p9}
\pnode(2.08, 4.52){p10}
\pnode(5.30, 1.57){p11}
\pnode(4.72, 1.76){p12}
\pnode(2.21, 2.60){p13}
\pnode(2.78, 2.41){p14}
\pnode(6.99, 4.40){p15}
\pnode(7.23, 3.98){p16}
\pnode(6.63, 2.17){p17}
\pnode(6.89, 2.46){p18}
\pnode(6.89, 2.46){p19}
\pnode(6.57, 3.10){p20}
\pnode(5.99, 2.37){p21}
\pnode(2.04, 5.62){p22}
\pnode(2.53, 5.86){p23}
\pnode(1.52, 3.88){p24}
\pnode(1.48, 4.21){p25}
\pnode(1.48, 4.21){p26}
\pnode(2.08, 4.52){p27}
\pnode(2.12, 3.67){p28}
\pnode(4.06, 2.12){p29}
\pnode(3.52, 2.31){p30}

\pnode(5.89, 4.31){p31}
\pnode(3.48, 3.21){p32}

}
}


\def\cones{
	\pnode(0,0){O}
	\pnode(100,0){I}
	\pnode([nodesep=100,angle=72]{I}O){II}\uput[270](II){$II$}
	\pnode([nodesep=100,angle=72]{II}O){III}
	\pnode([nodesep=100,angle=72]{III}O){IV}
	\pnode([nodesep=100,angle=72]{IV}O){V}
		
	\psline[linewidth=0pt, linecolor=lightgray](O)(I)
	\psline[linewidth=0pt, linecolor=lightgray](O)(II)
	\psline[linewidth=0pt, linecolor=lightgray](O)(III)
	\psline[linewidth=0pt, linecolor=lightgray](O)(IV)
	\psline[linewidth=0pt, linecolor=lightgray](O)(V)
}



\uput{0}[0](p1){\cones}
\uput{0}[0](p2){\cones}
\uput{0}[0](p3){\cones}
\uput{0}[0](p4){\cones}
\uput{0}[0](p5){\cones}
\uput{0}[0](p6){\cones}
\uput{0}[0](p7){\cones}
\uput{0}[0](p8){\cones}
\uput{0}[0](p9){\cones}
\uput{0}[0](p10){\cones}
\uput{0}[0](p11){\cones}
\uput{0}[0](p12){\cones}
\uput{0}[0](p13){\cones}
\uput{0}[0](p14){\cones}
\uput{0}[0](p15){\cones}
\uput{0}[0](p16){\cones}
\uput{0}[0](p17){\cones}
\uput{0}[0](p18){\cones}
\uput{0}[0](p22){\cones}
\uput{0}[0](p23){\cones}
\uput{0}[0](p24){\cones}
\uput{0}[0](p25){\cones}

	\psline(p2)(p4)(p6)(p3)(p29)(p4)(p30)(p3)(p5)(p4)(p13)(p9)(p14)(p30)(p5)(p6)(p29)(p12)(p7)(p11)(p3)(p1)
	\psline(p4)(p14)
	\psline(p3)(p12)
	\psline(p4)(p10)(p24)(p26)(p25)(p9)(p1)(p22)(p10)(p23)(p27)(p10)
	\psline(p9)(p28)
	\psline(p3)(p8)(p17)(p19)(p18)(p7)(p2)(p16)(p8)(p15)(p20)(p8)
	\psline(p7)(p21)
	\psline(p2)(p15)(p16)(p20)(p18)(p19)(p17)(p21)(p11)(p12)(p6)(p5)(p14)(p13)(p28)(p24)(p26)(p25)(p27)(p22)(p23)(p1)
	
	\psline[linewidth=0.2pt,linestyle=dashed](p4)(p1)(p2)(p3)
	\psline[linewidth=0.2pt,linestyle=dashed](p2)(p20)(p21)(p3)(p4)(p28)(p27)(p1)
	
	\psline[linewidth=1.2pt,linecolor=cyan](p1)(p23)(p10)(p4)(p2)
	\psline[linewidth=1.2pt,linecolor=red](p1)(p9)(p25)(p26)(p24)(p10)(p4)(p2)
	\psline[linewidth=1.2pt,linecolor=green](p1)(p9)(p13)(p4)(p2)
	\psline[linewidth=1.2pt,linecolor=orange](p1)(p3)(p5)(p4)(p2)
	
	\psset{offset=0.6pt}
	\ncline[linewidth=1.2pt, linecolor=green]{p2}{p4}
	\ncline[linewidth=1.2pt, linecolor=red]{p4}{p2}
	\psset{offset=1.8pt}
	\ncline[linewidth=1.2pt, linecolor=orange]{p31}{p32}
	\ncline[linewidth=1.2pt, linecolor=cyan]{p32}{p31}
	
	\psset{offset=0.6pt}
	\ncline[linewidth=1.2pt, linecolor=cyan]{p10}{p4}
	\ncline[linewidth=1.2pt, linecolor=red]{p4}{p10}
	
	\psset{offset=0.6pt}
	\ncline[linewidth=1.2pt, linecolor=green]{p1}{p9}
	\ncline[linewidth=1.2pt, linecolor=red]{p9}{p1}

	\pscircle[linestyle=dashed, linecolor=purple](5.15,5.9){0.6}




\psdot(p1)
\psdot(p2)
\psdot(p3)
\psdot(p4)
\psdot(p5)
\psdot(p6)
\psdot(p7)
\psdot(p8)
\psdot(p9)
\psdot(p10)
\psdot(p11)
\psdot(p12)
\psdot(p13)
\psdot(p14)
\psdot(p15)
\psdot(p16)
\psdot(p17)
\psdot(p18)
\psdot(p19)
\psdot(p20)
\psdot(p21)
\psdot(p22)
\psdot(p23)
\psdot(p24)
\psdot(p25)
\psdot(p26)
\psdot(p27)
\psdot(p28)



\uput[-30](p2){$u$}
\uput[-90](p1){$v$}
\uput{0.3}[-80](p4){$w$}
\uput{0.3}[250](p3){$z$}

\end{pspicture*}
\begin{pspicture*}(0,0)(10,1)
\uput[-70](5,1){(a)}
\end{pspicture*}
\begin{pspicture*}(0,0)(10,8)

\scalebox{1.5}{

\rput{216}(4.5,8){

	
\pnode(3.29 ,5.05){p1}
\pnode(5.81 ,4.23){p2}
\pnode(4.59 ,2.75){p3}
\pnode(3.41 ,3.12){p4}
\pnode(3.59 ,2.03){p5}
\pnode(6.22 ,2.36){p6}
\pnode(6.50 ,2.76){p7}
\pnode(1.86 ,3.75){p8}
\pnode(1.86 ,4.25){p9}
\pnode(5.22 ,1.21){p10}
\pnode(4.87 ,1.38){p11}
\pnode(1.94 ,2.33){p12}
\pnode(2.38 ,2.18){p13}
\pnode(1.76 ,5.60){p14}
\pnode(7.00 ,4.30){p15}
\pnode(7.39 ,3.90){p16}
\pnode(6.63 ,2.29){p17}
\pnode(6.70 ,2.41){p18}
\pnode(1.50 ,4.08){p19}
\pnode(1.49 ,3.93){p20}
\pnode(2.38 ,5.80){p21}
\pnode(6.45 ,4.69){p22}
\pnode(6.81 ,2.76){p23}
\pnode(6.32 ,2.08){p24}
\pnode(1.62 ,4.42){p25}
\pnode(1.62 ,3.58){p26}
\pnode(3.03 ,5.80){p27}
\pnode(7.00 ,2.92){p28}
\pnode(3.80 ,1.95){p29}
\pnode(1.53 ,4.68){p30}
\pnode(6.73 ,2.31){p31}
\pnode(1.40 ,4.00){p32}
\pnode(4.28 ,1.85){p33}
\pnode(3.14 ,2.21){p34}

\pnode(5.81 ,4.23){p35}
\pnode(3.41 ,3.12){p36}

}
}


\def\cones{
	\pnode(0,0){O}
	\pnode(100,0){I}
	\pnode([nodesep=100,angle=72]{I}O){II}\uput[270](II){$II$}
	\pnode([nodesep=100,angle=72]{II}O){III}
	\pnode([nodesep=100,angle=72]{III}O){IV}
	\pnode([nodesep=100,angle=72]{IV}O){V}
		
	\psline[linewidth=0pt, linecolor=lightgray](O)(I)
	\psline[linewidth=0pt, linecolor=lightgray](O)(II)
	\psline[linewidth=0pt, linecolor=lightgray](O)(III)
	\psline[linewidth=0pt, linecolor=lightgray](O)(IV)
	\psline[linewidth=0pt, linecolor=lightgray](O)(V)
}


\uput{0}[0](p1){\cones}
\uput{0}[0](p2){\cones}
\uput{0}[0](p3){\cones}
\uput{0}[0](p4){\cones}
\uput{0}[0](p5){\cones}
\uput{0}[0](p6){\cones}
\uput{0}[0](p7){\cones}
\uput{0}[0](p8){\cones}
\uput{0}[0](p9){\cones}
\uput{0}[0](p10){\cones}
\uput{0}[0](p11){\cones}
\uput{0}[0](p12){\cones}
\uput{0}[0](p13){\cones}
\uput{0}[0](p14){\cones}
\uput{0}[0](p15){\cones}
\uput{0}[0](p16){\cones}
\uput{0}[0](p17){\cones}
\uput{0}[0](p18){\cones}
\uput{0}[0](p19){\cones}
\uput{0}[0](p20){\cones}
\uput{0}[0](p21){\cones}
\uput{0}[0](p22){\cones}
\uput{0}[0](p23){\cones}
\uput{0}[0](p24){\cones}
\uput{0}[0](p25){\cones}
\uput{0}[0](p26){\cones}
\uput{0}[0](p27){\cones}
\uput{0}[0](p28){\cones}
\uput{0}[0](p29){\cones}
\uput{0}[0](p30){\cones}

	\psline(p2)(p4)(p29)(p3)(p5)(p4)(p33)(p3)(p34)(p4)(p12)(p8)(p13)(p34)(p5)(p29)(p33)(p11)(p6)(p10)(p3)(p1)
	\psline(p4)(p9)(p20)(p32)(p19)(p8)(p1)(p14)(p30)(p21)(p25)(p9)(p30)(p27)(p21)(p9)
	\psline(p19)(p25)(p30)(p19)
	\psline(p8)(p26)(p12)(p20)(p26)(p13)
	\psline(p3)(p7)(p17)(p31)(p18)(p6)(p2)(p16)(p28)(p15)(p23)(p7)(p28)(p22)(p15)(p7)
	\psline(p18)(p23)(p28)(p18)
	\psline(p6)(p24)(p10)(p17)(p24)(p11)
	\psline(p2)(p22)(p15)(p16)(p28)(p31)(p10)(p11)(p29)(p5)(p13)(p12)(p32)(p30)(p14)(p21)(p27)(p1)
	
	\psline[linewidth=0.2pt,linestyle=dashed](p13)(p4)(p1)(p2)(p3)(p11)
	\psline[linewidth=0.2pt,linestyle=dashed](p2)(p7)(p6)(p3)(p4)(p8)(p9)(p1)
	\psline[linewidth=0.2pt,linestyle=dashed](p2)(p15)
	\psline[linewidth=0.2pt,linestyle=dashed](p1)(p21)
	\psline[linewidth=0.2pt,linestyle=dashed](p19)(p9)(p14)
	\psline[linewidth=0.2pt,linestyle=dashed](p18)(p7)(p16)
	\psline[linewidth=0.2pt,linestyle=dashed](p8)(p20)
	\psline[linewidth=0.2pt,linestyle=dashed](p6)(p17)
	
	\psline[linewidth=1.2pt,linecolor=cyan](p1)(p27)(p30)(p9)(p4)(p2)
	\psline[linewidth=1.2pt,linecolor=red](p1)(p8)(p19)(p32)(p20)(p9)(p4)(p2)
	\psline[linewidth=1.2pt,linecolor=green](p1)(p8)(p12)(p4)(p2)
	\psline[linewidth=1.2pt,linecolor=orange](p1)(p3)(p5)(p4)(p2)
	
	\psset{offset=0.6pt}
	\ncline[linewidth=1.2pt, linecolor=green]{p2}{p4}
	\ncline[linewidth=1.2pt, linecolor=red]{p4}{p2}
	\psset{offset=1.8pt}
	\ncline[linewidth=1.2pt, linecolor=orange]{p35}{p36}
	\ncline[linewidth=1.2pt, linecolor=cyan]{p36}{p35}
	
	\psset{offset=0.6pt}
	\ncline[linewidth=1.2pt, linecolor=cyan]{p9}{p4}
	\ncline[linewidth=1.2pt, linecolor=red]{p4}{p9}
	
	\psset{offset=0.6pt}
	\ncline[linewidth=1.2pt, linecolor=green]{p1}{p8}
	\ncline[linewidth=1.2pt, linecolor=red]{p8}{p1}




\psdot(p1)
\psdot(p2)
\psdot(p3)
\psdot(p4)
\psdot(p5)
\psdot(p6)
\psdot(p7)
\psdot(p8)
\psdot(p9)
\psdot(p10)
\psdot(p11)
\psdot(p12)
\psdot(p13)
\psdot(p14)
\psdot(p15)
\psdot(p16)
\psdot(p17)
\psdot(p18)
\psdot(p19)
\psdot(p20)
\psdot(p21)
\psdot(p22)
\psdot(p23)
\psdot(p24)
\psdot(p25)
\psdot(p26)
\psdot(p27)
\psdot(p28)
\psdot(p29)
\psdot(p30)
\psdot(p31)
\psdot(p32)
\psdot(p33)
\psdot(p34)



\uput[-30](p2){$u$}
\uput[-90](p1){$v$}
\uput{0.3}[-80](p4){$w$}
\uput{0.3}[240](p3){$z$}

\end{pspicture*}
\begin{pspicture*}(0,0)(10,1)
\uput[-70](5,1){(b)}
\end{pspicture*}
\caption{(a) shows that the fractal growth is limited by collision of branches (in the circled area), lowering the stretch factor to 2.66. In (b), the stretch factor is increased to 2.87 by adjusting the shape of fractal to equalize the lengths of the shortest paths between $u$ and $v$. The shortest paths between $u$ and $v$ in the left side of the figure are shown as colors paths.}\label{fig:frac-4}
\end{center}
\end{figure}


%% file: fig-yy.tex

\begin{figure}[tbhp]
\begin{center}\psset{unit=18pt}
\begin{pspicture*}(-1,-1)(12.8,5)

	\pnode(0,0){a}\uput[180](a){$a$}
	\pnode(1.3,3.5){b}\uput[180](b){$b$}
	\pnode([nodesep=13,angle=1]{a}a){a0}
	\pnode([nodesep=13,angle=-1]{b}b){b0}
	\pnode([nodesep=2.6]{a0}a){a1}
	\pnode([nodesep=2.6]{a0}a1){a2}
	\pnode([nodesep=2.6]{a0}a2){a3}
	\pnode([nodesep=2.6]{a0}a3){a4}
	\pnode([nodesep=2.6]{b0}b){b1}
	\pnode([nodesep=2.6]{b0}b1){b2}
	\pnode([nodesep=2.6]{b0}b2){b3}
	\pnode([nodesep=2.6]{b0}b3){b4}

\def\cones{
	\pnode(0,0){O}
	\pnode(100,0){I}
	\pnode([nodesep=100,angle=72]{I}O){II}\uput[270](II){$II$}
	\pnode([nodesep=100,angle=72]{II}O){III}
	\pnode([nodesep=100,angle=72]{III}O){IV}
	\pnode([nodesep=100,angle=72]{IV}O){V}
		
	\psline[linecolor=lightgray](O)(I)
	\psline[linecolor=lightgray](O)(II)
	\psline[linecolor=lightgray](O)(III)
	\psline[linecolor=lightgray](O)(IV)
	\psline[linecolor=lightgray](O)(V)
}


\uput{0}[0](a
){\cones}
\uput{0}[0](b
){\cones}
\uput{0}[0](a0
){\cones}
\uput{0}[0](b0
){\cones}
\uput{0}[0](a1
){\cones}
\uput{0}[0](a2
){\cones}
\uput{0}[0](a3
){\cones}
\uput{0}[0](a4
){\cones}
\uput{0}[0](b1
){\cones}
\uput{0}[0](b2
){\cones}
\uput{0}[0](b3
){\cones}
\uput{0}[0](b4
){\cones}


	\psline[linewidth=1.5pt,linecolor=red](a)(a1)(a2)(a3)(a4)(a0)(b0)(b4)(b3)(b2)(b1)(b)

\psdot(a
)
\psdot(b
)
\psdot(a0
)
\psdot(b0
)
\psdot(a1
)
\psdot(a2
)
\psdot(a3
)
\psdot(a4
)
\psdot(b1
)
\psdot(b2
)
\psdot(b3
)
\psdot(b4
)

\pnode(12.1,1.75){d1}
\pnode([nodesep=0.3]{d1}d1){d2}
\pnode([nodesep=0.3]{d2}d2){d3}
\psdot[linecolor=red,dotstyle=*,dotscale=0.8](d1)
\psdot[linecolor=red,dotstyle=*,dotscale=0.8](d2)
\psdot[linecolor=red,dotstyle=*,dotscale=0.8](d3)

\end{pspicture*}
\caption{The initial steps of constructing a $YY_5$ graph with unbounded stretch factor. The pattern continues to the right. The gray lines are the boundaries of the cones.}\label{fig:yy}
\end{center}
\end{figure}
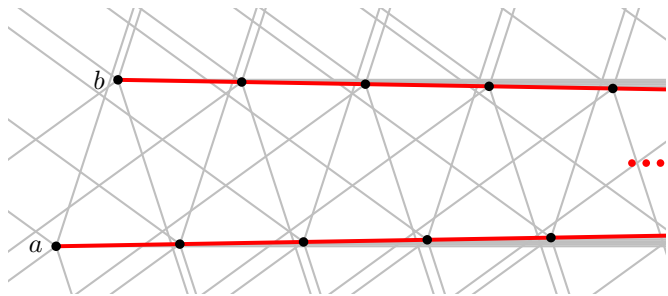

%% file: paper.bbl
\begin{thebibliography}{10}

\bibitem{personal}
L.~Barba, P.~Bose, M.~Damian, R.~Fagerberg, J.~O'Rourke, A.~van Renssen,
  P.~Taslakian, and S.~Verdonschot.
\newblock New and improved spanning ratios for {Yao} graphs.
\newblock {\em CoRR}, abs/1307.5829, 2013.

\bibitem{soda13}
M.~Bauer and M.~Damian.
\newblock An infinite class of sparse-yao spanners.
\newblock In {\em SODA}, pages 184--196, 2013.

\bibitem{y7}
P.~Bose, M.~Damian, K.~Dou\"{\i}eb, J.~O'Rourke, B.~Seamone, M.~Smid, and
  S.~Wuhrer.
\newblock $\pi/2$-angle {Y}ao graphs are spanners.
\newblock {\em CoRR}, abs/1001.2913, 2010.

\bibitem{y4}
P.~Bose, M.~Damian, K.~Dou\"{\i}eb, J.~O'Rourke, B.~Seamone, M.~Smid, and
  S.~Wuhrer.
\newblock $\pi/2$-angle {Y}ao graphs are spanners.
\newblock {\em Int. J. Comput. Geometry Appl.}, 22(1):61--82, 2012.

\bibitem{eucg}
M.~Damian, N.~Molla, and V.~Pinciu.
\newblock Spanner properties of $\pi/2$-angle {Y}ao graphs.
\newblock In {\em Proc. of the 25th European Workshop on Computational
  Geometry}, pages 21--24, 2009.

\bibitem{y6}
M.~Damian and K.~Raudonis.
\newblock Yao graphs span theta graphs.
\newblock In {\em Proceedings of the 4th international conference on
  Combinatorial optimization and applications - Volume Part II}, COCOA'10,
  pages 181--194, Berlin, Heidelberg, 2010. Springer-Verlag.

\bibitem{joco}
I.A. Kanj and G.~Xia.
\newblock On certain geometric properties of the {Yao-Yao} graphs.
\newblock In {\em COCOA}, pages 223--233, 2012.

\bibitem{early}
W.L. Keng and G.~Xia.
\newblock The {Y}ao graph $y_5$ is a spanner.
\newblock {\em CoRR}, abs/1307.5030, 2013.

\bibitem{molla}
N.~Molla.
\newblock Yao spanners for wireless ad hoc networks.
\newblock M.S. Thesis, Department of Computer Science, Villanova University,
  December 2009.

\bibitem{yao}
A.~C.-C. Yao.
\newblock On constructing minimum spanning trees in $k$-dimensional spaces and
  related problems.
\newblock {\em SIAM Journal on Computing}, 11(4):721--736, 1982.

\end{thebibliography}
